\documentclass[11pt]{extarticle}

\usepackage{amssymb,amsfonts,amsthm}
\usepackage{amsmath}
\usepackage{ytableau}
\usepackage[top=1in, bottom=1.25in, left=0.75in, right=0.75in]{geometry}
\usepackage{authblk}
\usepackage{hyperref}

\newtheorem{theorem}{Theorem}[section]
\newtheorem{lemma}[theorem]{Lemma}
\newtheorem{corollary}[theorem]{Corollary}

\theoremstyle{definition}
\newtheorem{definition}[theorem]{Definition}
\newtheorem{proposition}[theorem]{Proposition}

\theoremstyle{remark}
\newtheorem{remark}[theorem]{Remark}

\numberwithin{equation}{section}

\newcommand{\C}{\mathbb{C}}

\newcommand{\Sn}{\mathfrak S}

\DeclareMathOperator{\Span}{span}

\newcommand{\ip}[2]{\left\langle #1,\,#2\right\rangle}

\newcommand{\cV}{\mathcal V}
\newcommand{\cH}{\mathcal H}
\newcommand{\cE}{\mathcal E}
\newcommand{\cW}{\mathcal W}

\newcommand{\ket}[1]{\left|#1\right\rangle}

\newcommand{\proj}[1]{\left|#1\right\rangle\!\left\langle#1\right|}
\newcommand{\TIG}{T^{\mathrm{IG}}}
\newcommand{\TiSWAP}{T^{\mathrm{iSWAP}}}
\newcommand{\spec}{\operatorname{spec}}

\newcommand{\lamstar}{\lambda_{\!*}}

\hypersetup{
  colorlinks=true,
  linkcolor=blue,
  citecolor=blue,
  urlcolor=blue
}

\begin{document}

\title{Spectral gaps of ironed two-qubit gadgets matching the iSWAP gap}

\author{Yanying Liang\thanks{College of Mathematics and Informatics, South China Agricultural University, Guangzhou, 510640, China. Email: yyl2022@scau.edu.cn}
\thanks{Institute of Advanced Intelligence and Computing (IAIC), Agency for Science, Technology and Research (A*STAR), 1 Fusionopolis Way, \#16-16 Connexis, Singapore 138632, Republic of Singapore}
\qquad Haoran Zhu\thanks{Division of Mathematical Sciences, Nanyang Technological University, Singapore 637371, Singapore. Email: zhuh0031@e.ntu.edu.sg}}

\date{}

\maketitle

\begin{abstract}
We prove that every ironed two-qubit gadget whose KAK-derived parameter satisfies
$a=5/9$ has, on the complete graph $K_n$ with $n\geqslant 5$, the same
second-moment spectral gap as the iSWAP gadget.  
The central step is a representation-theoretic localisation theorem: the largest strictly negative eigenvalue of the associated $\Sn_n$-invariant operator always occurs in the highest-spin $\mathrm{SU}(2)$ summand. A local positive-semidefinite decomposition separates every spin sector except the two highest. This settles a conjecture of Kong, Li, and Liu.  
\end{abstract}



\section{Introduction}

Exact unitary designs reproduce prescribed Haar moments, while approximate
unitary designs reproduce them up to controlled error.  They therefore
provide finite or efficiently generated substitutes for Haar-random
unitaries in applications that depend on finitely many moments.  Early
systematic work on exact and approximate unitary designs includes the
constructions of Dankert, Cleve, Emerson and Livine~\cite{DankertCleveEmersonLivine}
and Gross, Audenaert and Eisert~\cite{GrossAudenaertEisert}.  Random quantum
circuits give a particularly natural source of approximate designs: rather
than sampling from a global unitary group, one composes local random gates.
Quantitative convergence results for these circuits were established for
second moments by Harrow and Low~\cite{HarrowLow}, with a correction to a key
argument by Diniz and Jonathan~\cite{DinizJonathan}; for higher moments by
Brand\~ao, Harrow and Horodecki~\cite{BrandaoHarrowHorodecki}; and subsequently
with substantially improved depth bounds by Haferkamp~\cite{Haferkamp}.

For a Hermitian moment operator, the rate of convergence is governed by a
spectral gap.  This observation turns a problem about random circuits into
a problem in finite-dimensional spectral theory.  The resulting matrices
are nevertheless large: even at the second moment, the ambient space grows
exponentially with the number of qubits.  Obtaining a sharp or exact gap
therefore requires exploiting the symmetries of the circuit architecture
and of the local gate ensemble.

The moment-operator viewpoint has been developed through several
complementary Hamiltonian and statistical-mechanical correspondences.
Brown and Viola~\cite{BrownViola} mapped arbitrary fixed moments of all-to-all random
circuits to collective-spin Hamiltonians, whilst
Hunter-Jones~\cite{HunterJones} expressed design formation through a lattice partition
function.  Geometry-sensitive depth estimates were
obtained by Harrow and Mehraban~\cite{HarrowMehraban}, and improved gap
bounds for large local dimension and all-to-all interactions were proved
by Haferkamp and Hunter-Jones~\cite{HaferkampHunterJones}.  More recently,
Deneris et al.~\cite{DenerisBermejoBracciaCincioCerezo} obtained exact
second-moment gaps for one-dimensional Haar-random circuits, Yada et
al.~\cite{YadaSuzukiMitsuhashiYoshioka} proved that broad classes of
non-Haar random circuits form designs with only a system-size-independent
constant-factor overhead relative to their Haar-random counterparts, and
Baer and Haah~\cite{BaerHaah} established constant spectral gaps for
several other random-unitary circuit models.  These works concern different
architectures or local ensembles from the complete-graph ironed-gadget
model studied here.  In related work in preparation, Liang and
Zhu~\cite{LiangZhu2026} investigate the second-moment spectral gap for
random circuits generated by two-qubit gates.  These results reinforce a
recurring principle: sharp convergence estimates depend not only on
positivity or comparison inequalities, but also on identifying which
representation sector contains the decisive low-lying excitation.

Kong, Li, and Liu~\cite{KongLiLiu} introduced the \textbf{ironed gadget}
model to compare two-qubit entangling gates after averaging away local
one-qubit degrees of freedom.  A fixed two-qubit unitary is preceded and
followed by independent Haar-random one-qubit gates.  Since two-qubit
unitaries are classified up to local gates by their KAK coordinates
\cite{KhanejaBrockettGlaser,ZhangValaSastryWhaley}, the second-moment
operator of the gadget is described by three scalar quantities, denoted by
$a,b,c$, subject to the relation $b=(1-a)/2$.  Placing the same gadget on a
graph then produces a local random circuit whose one-step moment operator
is the average over the edges.

On the complete graph, the permutation symmetry is strong enough to make
Schur--Weyl duality decisive.  Kong, Li, and Liu showed that the restriction
of the second-moment operator to its natural invariant subspace has the
form
\begin{equation*}\label{eq:intro-decomposition}
 I+(1-a)A_n+cB_n,
\end{equation*}
where $A_n$ and $B_n$ commute with the action of the symmetric group
$\Sn_n$.  The operator $B_n$ is central on each symmetric-group isotypic
component.  The operator $A_n$ is more delicate: after Schur--Weyl
decomposition, it becomes a family of tridiagonal matrices indexed by total
spin.

The iSWAP, $B$-gate, and CNOT families all satisfy $a=5/9$, but they have
different values of $c$.  For the complete graph \(K_n\) with \(n\geqslant5\),
Kong, Li, and Liu conjectured that every ironed gadget with KAK parameter
\(a=5/9\) has the same second-moment spectral gap as the iSWAP gadget
\cite[Conjecture~2]{KongLiLiu}. In the current paper, we confirm this conjecture; see the following theorem.



\begin{theorem}
\label{thm:intro-main}
Let $n\geqslant 5$, and let $\TIG_{2,K_n}$ be the second-moment operator of an
ironed two-qubit gadget whose KAK parameter is $a=5/9$.  Then
\begin{equation}\label{eq:intro-main}
 \Delta\!\left(\TIG_{2,K_n}\right)
 =\Delta\!\left(\TiSWAP_{2,K_n}\right).
\end{equation}
\end{theorem}

The obstruction identified in~\cite{KongLiLiu} is precise.  The authors
observed that their conjecture would follow from a strengthening of their
highest-spin lemma: one must show that the largest strictly negative
eigenvalue of $A_n$ occurs in the highest-spin summand for every $n\geqslant
5$.  Their estimates prove the required statement after the additional
iSWAP contribution from $B_n$ is included, but do not compare $A_n$
itself sharply enough.  Moreover, numerical evidence in~\cite{KongLiLiu}
indicates that the two highest-spin blocks approach one another as $n$
grows, so a coarse operator-norm estimate loses the relevant second-order
separation.

We prove the required strengthening.  To state it, write the Schur--Weyl
decomposition of the effective two-dimensional local space as
\begin{equation*}\label{eq:intro-SW}
 \cV^{\otimes n}
 \cong
 \bigoplus_{r=0}^{\lfloor n/2\rfloor}
 \cW_{n,r}\otimes S^{(n-r,r)},
\end{equation*}
where $\cW_{n,r}$ has total spin $(n-2r)/2$.  Put $H_n=-A_n$, write
$H_{n,r}$ for its restriction to $\cW_{n,r}$, and write
$\lambda_{\min}^{+}(M)$ for the smallest positive eigenvalue of a
positive-semidefinite matrix $M$.  Our spectral localisation theorem is as
follows.

\begin{theorem}\label{thm:intro-localisation}
For every $n\geqslant 5$ and every $1\leqslant r\leqslant\lfloor n/2\rfloor$,
\begin{equation}\label{eq:intro-localisation}
 \lambda_{\min}^{+}(H_{n,0})<\lambda_{\min}(H_{n,r}).
\end{equation}
Equivalently, the largest strictly negative eigenvalue of $A_n$ occurs in
the highest-spin summand $\cW_{n,0}$.
\end{theorem}

Firstly, the two-site term of $H_n$ admits a rank-two positive-semidefinite
decomposition.  One of its summands is the antisymmetric projection, and
this gives the uniform bound
\begin{equation*}\label{eq:intro-lowspin}
 \lambda_{\min}(H_{n,r})
 \geqslant \frac{3r(n-r+1)}{2n(n-1)}.
\end{equation*}
It follows immediately that every block with $r\geqslant 2$ has
ground-state energy at least $3/n$.  This leaves only the comparison
between $r=0$ and $r=1$.

Secondly, these two blocks are explicit Jacobi matrices. 
A two-dimensional Rayleigh--Ritz calculation, carried out after projecting
away the non-trivial zero mode of the highest-spin block, proves that
\begin{equation*}
 \lambda_{\min}^{+}(H_{n,0})<\frac{3n-5}{n(n-1)}
 \qquad (n\geqslant6).
\end{equation*}
Thirdly, a Sturm argument proves that
\begin{equation*}
 \lambda_{\min}(H_{n,1})>\frac{3n-5}{n(n-1)}
 \qquad (n\geqslant9).
\end{equation*}
The latter reduces to positivity of the leading principal minors of a
tridiagonal matrix.  
For $5\leqslant n\leqslant 8$, exact
rational Sturm sequences provide finite certificates.

It is noteworthy that, on a structural level, our findings are closely aligned in spirit with Aldous-type spectral-gap questions: a large operator arising from representation theory exhibits its crucial low-lying eigenvalue within a particular small representation (specifically, the standard representation). 
Aldous' conjecture for the interchange process was established by Caputo, Liggett, and Richthammer~\cite{CaputoLiggettRichthammer}.

Once Theorem~\ref{thm:intro-localisation} has been established, the parameter $c$ no longer plays a problematic role. 
At $a=5/9$ the admissible KAK region forces
$0\leqslant c\leqslant 1/3$, whilst $B_n$ vanishes on the highest-spin
block and is strictly negative on every lower-spin block.  Consequently,
increasing $c$ only lowers the competing eigenvalues.  The second-largest
eigenvalue on $\cV^{\otimes n}$ is therefore independent of $c$.  A
complementary-subspace reduction of Kong, Li, and Liu then
identifies it with the second-largest eigenvalue of the full moment
operator and shows that it, rather than the bottom of the spectrum,
determines the gap.

The paper is organised as follows.  In Section~\ref{sec:setup}, we review
second-moment operators and ironed gadgets, determine the admissible range
of $c$ when $a=5/9$, and record the complete-graph Schur--Weyl reduction.
In Section~\ref{sec:spin-blocks}, we derive the Jacobi matrices for all
spin sectors and use the local positive-semidefinite decomposition to
separate the sectors with $r\geqslant 2$.  In Section~\ref{sec:critical}, we
compare the two remaining Jacobi matrices by variational and Sturm
arguments and prove Theorem~\ref{thm:intro-localisation}.  In
Section~\ref{sec:main-proof}, we return to the moment operator and prove
Theorem~\ref{thm:intro-main}.  In Appendix~\ref{app:small}, we give the
exact finite-dimensional Sturm certificates used for
$5\leqslant n\leqslant 8$.

\section{Moment operators and the complete-graph reduction}\label{sec:setup}

\subsection{Second moments and spectral gaps}

Let
$\cH_n=(\C^2)^{\otimes n}$
be the Hilbert space of an $n$-qubit system.  If $\cE$ is an ensemble of
unitaries on $\cH_n$, its \textbf{second-moment operator} is
\begin{equation}\label{eq:moment-operator}
 T_2^{\cE}
 =\mathbb E_{U\sim\cE}
 \bigl(U^{\otimes 2}\otimes\overline U^{\otimes 2}\bigr).
\end{equation}
Equivalently, after vectorisation, this is the averaging operator on
$\operatorname{End}(\cH_n^{\otimes 2})$ induced by conjugation with
$U^{\otimes 2}$.  The \textbf{Haar second-moment operator} is the orthogonal
projection onto the corresponding commutant.


\begin{definition}\label{def:lambda-star}
Let $T$ be a Hermitian contraction whose unit eigenspace is non-zero and
whose spectrum contains at least one eigenvalue strictly below $1$.  We
write
\begin{equation}\label{eq:lambda-star}
 \lamstar(T)=\max\{\lambda\in\spec(T):\lambda<1\}
\end{equation}
for its largest eigenvalue strictly below $1$.  Its moment-operator
\textbf{spectral gap} is
\begin{equation}\label{eq:spectral-gap}
 \Delta(T)
 =1-\max\bigl\{\lamstar(T),\,|\lambda_{\min}(T)|\bigr\}.
\end{equation}
\end{definition}


\subsection{Ironed two-qubit gadgets}

An \textbf{ironed gadget} is obtained by placing a fixed two-qubit unitary between
two independent layers of Haar-random one-qubit gates.  More precisely,
if $U\in\mathrm U(4)$ is fixed, the unitary applied to a chosen pair of
qubits has the form
\begin{equation*}\label{eq:ironed-gadget}
 (A_1\otimes A_2)U(B_1\otimes B_2),
\end{equation*}
where $A_1,A_2,B_1,B_2$ are independent and Haar distributed in
$\mathrm U(2)$.  Local unitary factors in $U$ are absorbed by the Haar
variables.  We may therefore take the non-local part of $U$ in the KAK
form
\begin{equation}\label{eq:KAK}
 U(k_x,k_y,k_z)
 =\exp\!\left(
 i(k_x\sigma_x\otimes\sigma_x
  +k_y\sigma_y\otimes\sigma_y
  +k_z\sigma_z\otimes\sigma_z)
 \right),
\end{equation}
where $\sigma_x,\sigma_y,\sigma_z$ are the Pauli matrices.

For the second moment, the one-qubit Haar projector has a two-dimensional
range.  Let $I$ and $\mathsf F$ denote respectively the identity and the
swap on $(\C^2)^{\otimes 2}$.  With respect to the Hilbert--Schmidt inner
product, set
\begin{equation*}\label{eq:u0-u1}
 \mathbf u_0=\frac12 I,
 \qquad
 \mathbf u_1=\frac1{\sqrt3}\left(\mathsf F-\frac12 I\right),
 \qquad
 \cV=\Span\{\mathbf u_0,\mathbf u_1\}.
\end{equation*}
The local one-qubit Haar twirls are orthogonal projections onto $\cV$.
Consequently, the two-qubit gadget moment operator has range contained in
$\cV\otimes\cV$ and vanishes on its orthogonal complement.  The ordered
basis
\begin{equation*}\label{eq:local-basis}
 \mathbf u_0\otimes \mathbf u_0,
 \quad \mathbf u_0\otimes \mathbf u_1,
 \quad \mathbf u_1\otimes \mathbf u_0,
 \quad \mathbf u_1\otimes \mathbf u_1
\end{equation*}
identifies $\cV\otimes\cV$ with $\C^2\otimes\C^2$.  In this basis the
restriction of the local second-moment operator is~\cite[Eq.~(33)]{KongLiLiu}
\begin{equation}\label{eq:local-matrix}
 \left.\TIG_2\right|_{\cV\otimes\cV}=
 \begin{pmatrix}
 1&0&0&0\\
 0&1-c-3b&c&\sqrt3b\\
 0&c&1-c-3b&\sqrt3b\\
 0&\sqrt3b&\sqrt3b&a
 \end{pmatrix},
\end{equation}
where
\begin{align}
 a&=\frac19\bigl(6+xy+xz+yz\bigr),\label{eq:a-parameter}\\
 b&=\frac1{18}\bigl(3-xy-xz-yz\bigr)=\frac{1-a}{2},\label{eq:b-parameter}\\
 c&=\frac1{12}\bigl(3+xy+xz+yz-2(x+y+z)\bigr),\label{eq:c-parameter}
\end{align}
with
\begin{equation*}\label{eq:xyz}
 x=\cos(4k_x),\qquad y=\cos(4k_y),\qquad z=\cos(4k_z).
\end{equation*}

The function $xy+xz+yz$ is affine in each variable, so its minimum on
the cube $[-1,1]^3$ is attained at a vertex and equals $-1$.  Thus
$a\geqslant5/9$, and the boundary value $a=5/9$ is precisely the one singled
out in~\cite{KongLiLiu}.  The following observation will be
important throughout.

\begin{lemma}\label{lem:c-range}
Suppose that $a=5/9$.  Then
\begin{equation}\label{eq:c-range}
 0\leqslant c\leqslant\frac13.
\end{equation}
Moreover, the two non-unit eigenvalues of the matrix in
\eqref{eq:local-matrix}, regarded as an unordered pair, are
\begin{equation}\label{eq:local-nonunit-spectrum}
 -\frac19
 \qquad\text{and}\qquad
 \frac13-2c.
\end{equation}
In particular,
\begin{equation}\label{eq:local-min-lower}
 \lambda_{\min}(\TIG_2)\geqslant-\frac13.
\end{equation}
\end{lemma}

\begin{proof}
The equation $a=5/9$ is equivalent to
\begin{equation*}\label{eq:pair-sum-minus-one}
 xy+xz+yz=-1.
\end{equation*}
Substitution into \eqref{eq:c-parameter} gives
\begin{equation*}\label{eq:c-sum}
 c=\frac{1-(x+y+z)}6.
\end{equation*}
Since $x,y,z\in[-1,1]$, we have
\begin{equation*}\label{eq:sum-square}
 (x+y+z)^2
 =x^2+y^2+z^2+2(xy+xz+yz)
 =x^2+y^2+z^2-2
 \leqslant 1.
\end{equation*}
Thus $-1\leqslant x+y+z\leqslant 1$, and \eqref{eq:c-range} follows.

When $a=5/9$, equation \eqref{eq:b-parameter} gives $b=2/9$.
A direct determinant calculation in \eqref{eq:local-matrix} yields
\begin{equation*}\label{eq:local-characteristic}
 \det(\lambda I-\TIG_2)
 =\frac1{27}(\lambda-1)^2(9\lambda+1)(6c+3\lambda-1).
\end{equation*}
This proves \eqref{eq:local-nonunit-spectrum}.  Since
$1/3-2c\in[-1/3,1/3]$, the lower bound
\eqref{eq:local-min-lower} follows.
\end{proof}


It is useful to separate the two parameters 
 at the two-site
level.  Let $\mathsf a$ and $\mathsf b$ be the following operators on
$\cV\otimes\cV$:
\begin{equation}\label{eq:local-a-b}
 \mathsf a=
 \begin{pmatrix}
 0&0&0&0\\
 0&-\frac32&0&\frac{\sqrt3}{2}\\
 0&0&-\frac32&\frac{\sqrt3}{2}\\
 0&\frac{\sqrt3}{2}&\frac{\sqrt3}{2}&-1
 \end{pmatrix},
 \qquad
 \mathsf b=
 \begin{pmatrix}
 0&0&0&0\\
 0&-1&1&0\\
 0&1&-1&0\\
 0&0&0&0
 \end{pmatrix}.
\end{equation}
Equation \eqref{eq:local-matrix} is equivalently
\begin{equation*}\label{eq:local-affine-decomposition}
 \left.\TIG_2\right|_{\cV\otimes\cV}
 =I+(1-a)\mathsf a+c\mathsf b.
\end{equation*}
The second matrix is $\mathrm{SWAP}-I$ on the effective tensor product.
This affine decomposition clearly shows that the scalar $a$ governs a fixed, non-central interaction, while $c$ governs a permutation component.

For the complete graph $K_n$, one chooses an edge uniformly and applies
the gadget to its endpoints.  Put
\begin{equation*}\label{eq:Nn}
 N_n=n(n-1)=2|E(K_n)|.
\end{equation*}
The restriction of the resulting moment operator to $\cV^{\otimes n}$
is; see also~\cite[Eq.~(113)]{KongLiLiu}
\begin{equation}\label{eq:complete-restricted}
 \TIG_{2,K_n}\big|_{\cV^{\otimes n}}
 =I+(1-a)A_n+cB_n.
\end{equation}
To describe $A_n$ and $B_n$, let $X$ and $Z$ be the Pauli matrices on the
effective space $\cV$, in the ordered basis $(\mathbf u_0,\mathbf u_1)$,
and define the collective operators
\begin{equation*}\label{eq:collective}
 \rho(X)=\sum_{i=1}^n X_i,
 \qquad
 \rho(Z)=\sum_{i=1}^n Z_i.
\end{equation*}
Then
\begin{align}
 A_n={}&\frac2{N_n}\Bigg[
 \frac{\sqrt3}{4}\left(
 (n-1)\rho(X)
 -\frac12\bigl(\rho(X)\rho(Z)+\rho(Z)\rho(X)\bigr)
 \right)\notag\\
 &\hspace{4.6em}
 +\frac14\left((n-1)\rho(Z)+\rho(Z)^2\right)
 \Bigg]
 -\left(1+\frac1{2(n-1)}\right)I,
 \label{eq:A-definition}\\
 B_n={}&\frac2{N_n}\sum_{1\leqslant i<j\leqslant n}(ij)-I.
 \label{eq:B-definition}
\end{align}
Here $(ij)$ interchanges the $i$th and $j$th tensor factors of
$\cV^{\otimes n}$.

At the parameter value of interest, \eqref{eq:complete-restricted}
becomes
\begin{equation}\label{eq:restricted-a-five-nine}
 \TIG_{2,K_n}\big|_{\cV^{\otimes n}}
 =I+\frac49A_n+cB_n.
\end{equation}

\subsection{Schur--Weyl decomposition}

The commuting actions of $\mathrm{SU}(2)$ and $\Sn_n$ on
$\cV^{\otimes n}$ give the multiplicity-free Schur--Weyl decomposition
in the $\mathrm{SU}(2)\times\Sn_n$ sense
\begin{equation}\label{eq:Schur-Weyl}
 \cV^{\otimes n}
 \cong
 \bigoplus_{r=0}^{\lfloor n/2\rfloor}
 \cW_{n,r}\otimes S^{(n-r,r)}.
\end{equation}
Here $S^{(n-r,r)}$ is the Specht module indexed by the two-row partition
$(n-r,r)$, and $\cW_{n,r}$ is the irreducible $\mathrm{SU}(2)$ module of
spin
\begin{equation*}\label{eq:spin}
 j_{n,r}=\frac{n-2r}{2}
\end{equation*}
and dimension $n-2r+1$.  We refer to $r=0$ as the highest-spin sector and
to $r=1$ as the next-to-highest-spin sector.  Standard references for
Schur--Weyl duality and symmetric-group representations include
\cite{FultonHarris,GoodmanWallach,Sagan}.

The operator $A_n$ belongs to the enveloping algebra of the collective
$\mathrm{SU}(2)$ action, whereas $B_n$ belongs to the centre of the group
algebra of $\Sn_n$.  Hence both preserve every summand in
\eqref{eq:Schur-Weyl}.

\begin{lemma}\label{lem:B-scalar}
On the summand indexed by $r$ in \eqref{eq:Schur-Weyl}, the operator $B_n$
acts as
\begin{equation}\label{eq:B-scalar}
 B_n\big|_{\cW_{n,r}\otimes S^{(n-r,r)}}
 =-\frac{2r(n-r+1)}{N_n}I.
\end{equation}
In particular, $B_n$ vanishes in the highest-spin sector and is strictly
negative in every lower-spin sector.
\end{lemma}

\begin{proof}
The class sum of transpositions acts on $S^\lambda$ by the sum of the
contents of the Young diagram of $\lambda$.  For
$\lambda=(n-r,r)$, this scalar is
\begin{align*}
 \sum_{(i,j)\in\lambda}(j-i)
 &=\frac12\bigl((n-r)(n-r-1)+r(r-3)\bigr)=\frac12\bigl(N_n-2r(n-r+1)\bigr).
 \label{eq:content-sum}
\end{align*}
Substitution into \eqref{eq:B-definition} gives
\eqref{eq:B-scalar}.
\end{proof}


\section{Spin blocks of the complete-graph operator}\label{sec:spin-blocks}

We now study
\begin{equation}\label{eq:H-definition}
 H_n=-A_n.
\end{equation}
The sign is convenient because $H_n$ is positive semidefinite.  We shall
write $H_{n,r}$ for its restriction to the $\mathrm{SU}(2)$ factor
$\cW_{n,r}$ in \eqref{eq:Schur-Weyl}.

\subsection{Jacobi matrices for the spin sectors}

Fix $r$.  Choose the standard orthonormal weight basis of $\cW_{n,r}$ and
index it as
\begin{equation}\label{eq:weight-basis}
 \mathbf e_q^{(r)},\qquad r\leqslant q\leqslant n-r.
\end{equation}
The collective generators act by
\begin{align}
 \rho(Z)\mathbf e_q^{(r)}
 &= (n-2q)\mathbf e_q^{(r)},
 \label{eq:rhoZ-action}\\
 \rho(X)\mathbf e_q^{(r)}
 &=\sqrt{(q-r+1)(n-r-q)}\,\mathbf e_{q+1}^{(r)}+
 \sqrt{(q-r)(n-r-q+1)}\,\mathbf e_{q-1}^{(r)},
 \label{eq:rhoX-action}
\end{align}
where vectors outside the range in \eqref{eq:weight-basis} are understood
to be zero.

\begin{proposition}\label{prop:Jacobi-blocks}
In the basis \eqref{eq:weight-basis}, the matrix $H_{n,r}$ is tridiagonal.
Its diagonal entries are
\begin{equation}\label{eq:general-diagonal}
 \bigl(H_{n,r}\bigr)_{q,q}
 =\frac{q(3n-2q-1)}{N_n},
 \qquad r\leqslant q\leqslant n-r,
\end{equation}
and its upper off-diagonal entries are
\begin{equation}\label{eq:general-off-diagonal}
 \bigl(H_{n,r}\bigr)_{q,q+1}
 =-\frac{\sqrt3\,q
 \sqrt{(q-r+1)(n-r-q)}}{N_n},
 \qquad r\leqslant q<n-r.
\end{equation}
The lower off-diagonal entries are obtained by symmetry.
\end{proposition}

\begin{proof}
The part of \eqref{eq:A-definition} involving only $\rho(Z)$ gives
\begin{align*}
 \bigl(A_n\bigr)_{q,q}
 &=\frac{(n-1)(n-2q)+(n-2q)^2}{2N_n}
 -1-\frac1{2(n-1)}=-\frac{q(3n-2q-1)}{N_n},
\end{align*}
which proves \eqref{eq:general-diagonal} after changing sign.

The $(q,q+1)$ entry of $\rho(X)$ is
$\sqrt{(q-r+1)(n-r-q)}$, while the corresponding diagonal eigenvalues of
$\rho(Z)$ are $n-2q$ and $n-2q-2$.  Hence the first line of
\eqref{eq:A-definition} contributes
\begin{align*}
 &\frac2{N_n}\frac{\sqrt3}{4}
 \bigl((n-1)-(n-2q-1)\bigr)
 \sqrt{(q-r+1)(n-r-q)}
 =\frac{\sqrt3\,q\sqrt{(q-r+1)(n-r-q)}}{N_n}.
\end{align*}
Changing sign gives \eqref{eq:general-off-diagonal}.
\end{proof}

The two critical blocks will be used repeatedly.  In the highest-spin
sector, the basis is indexed by $0\leqslant q\leqslant n$.  The vector
$\mathbf e_0^{(0)}$ is an isolated zero eigenvector because the first off-diagonal
entry in \eqref{eq:general-off-diagonal} vanishes.  Removing it leaves an
$n\times n$ Jacobi matrix, which we denote by
\begin{equation*}\label{eq:H0-circle}
 H_{n,0}^{\circ}
 =H_{n,0}\big|_{\Span\{\mathbf e_1^{(0)},\ldots,\mathbf e_n^{(0)}\}}.
\end{equation*}
Its entries are
\begin{align}
 \bigl(H_{n,0}^{\circ}\bigr)_{q,q}
 &=\frac{q(3n-2q-1)}{N_n},
 &&1\leqslant q\leqslant n,
 \label{eq:H0-diagonal}\\
 \bigl(H_{n,0}^{\circ}\bigr)_{q,q+1}
 &=-\frac{\sqrt3\,q\sqrt{(q+1)(n-q)}}{N_n},
 &&1\leqslant q<n.
 \label{eq:H0-off-diagonal}
\end{align}
In the next-to-highest-spin sector, the basis is indexed by
$1\leqslant q\leqslant n-1$, and
\begin{align}
 \bigl(H_{n,1}\bigr)_{q,q}
 &=\frac{q(3n-2q-1)}{N_n},
 &&1\leqslant q\leqslant n-1,
 \label{eq:H1-diagonal}\\
 \bigl(H_{n,1}\bigr)_{q,q+1}
 &=-\frac{\sqrt3\,q\sqrt{q(n-q-1)}}{N_n},
 &&1\leqslant q<n-1.
 \label{eq:H1-off-diagonal}
\end{align}
Thus the diagonal sequences of the two matrices agree on their common
range.  Their off-diagonal entries are close, but not close enough for a
direct entrywise comparison after the highest-spin zero mode is removed.
For a square matrix $M$, write $M^{[k]}$ for its leading $k\times k$
principal submatrix.  Then
\begin{equation}\label{eq:critical-off-diagonal-ratio}
 \bigl(H_{n,1}\bigr)_{q,q+1}
 =\left(\bigl(H_{n,0}^{\circ}\bigr)^{[n-1]}\right)_{q,q+1}
 \sqrt{\frac{q(n-q-1)}{(q+1)(n-q)}}
 \qquad (1\leqslant q<n-1).
\end{equation}
The factor in \eqref{eq:critical-off-diagonal-ratio} lies strictly between
zero and one.  This suggests that the next-to-highest-spin block should
have the larger ground-state energy, but it does not yield an operator
inequality: the matrix
$H_{n,1}-\bigl(H_{n,0}^{\circ}\bigr)^{[n-1]}$ has zero diagonal and
non-zero positive off-diagonal entries, and is therefore indefinite.  
Since $H_{n,0}^{\circ}$ has eigenvalues
\begin{equation*}
 0<\lambda_{\min}^{+}(H_{n,0}^{\circ})<\cdots,
\end{equation*}
the interlacing property only guarantees that the smallest eigenvalue of
$\bigl(H_{n,0}^{\circ}\bigr)^{[n-1]}$ lies between $0$ and
$\lambda_{\min}^{+}(H_{n,0}^{\circ})$. Therefore, in Section~\ref{sec:critical} we introduce a common-threshold argument that is specifically constructed to bypass this zero-mode obstruction.

\subsection{A local positive-semidefinite decomposition}

We next establish positivity and separate all spin sectors with
$r\geqslant2$.  Let $\mathsf a_{ij}$ denote the copy of the first matrix in
\eqref{eq:local-a-b} acting on sites $i$ and $j$.  Its complete-graph
average is $A_n$.  Put
\begin{equation*}\label{eq:local-h}
 \mathsf h_{ij}=-\mathsf a_{ij}.
\end{equation*}
Suppressing the sites, the matrix of $\mathsf h=\mathsf h_{ij}$ in the
effective computational basis
\begin{equation*}\label{eq:effective-basis}
 \ket{00},\quad\ket{01},\quad\ket{10},\quad\ket{11}
\end{equation*}
is
\begin{equation}\label{eq:h-matrix}
 \mathsf h=
 \begin{pmatrix}
 0&0&0&0\\
 0&\frac32&0&-\frac{\sqrt3}{2}\\
 0&0&\frac32&-\frac{\sqrt3}{2}\\
 0&-\frac{\sqrt3}{2}&-\frac{\sqrt3}{2}&1
 \end{pmatrix}.
\end{equation}
Define the unit vectors
\begin{align}
 &{\psi}^{-}=\frac1{\sqrt2}(\ket{01}-\ket{10}),
 \label{eq:psi-minus}\\
 &{\eta}=\sqrt{\frac3{10}}(\ket{01}+\ket{10})
 -\sqrt{\frac25}\ket{11}.
 \label{eq:eta-vector}
\end{align}

\begin{lemma}\label{lem:local-PSD}
The local operator in \eqref{eq:h-matrix} has the decomposition
\begin{equation}\label{eq:h-decomposition}
 \mathsf h=\frac32\proj{{\psi}^{-}}+\frac52\proj{{\eta}}.
\end{equation}
Consequently,
\begin{equation}\label{eq:h-antisymmetric-bound}
 \mathsf h\succeq\frac32P^-,
 \qquad
 P^-=\frac{I-\mathrm{SWAP}}2.
\end{equation}
\end{lemma}

\begin{proof}
The vectors ${\psi}^{-}$ and ${\eta}$ are orthonormal.  Expanding the two
rank-one operators on the right-hand side of \eqref{eq:h-decomposition}
gives the matrix in \eqref{eq:h-matrix}.  Since
$P^-=\proj{{\psi}^{-}}$ on $\C^2\otimes\C^2$, equation
\eqref{eq:h-antisymmetric-bound} follows.
\end{proof}

Averaging \eqref{eq:h-antisymmetric-bound} over the complete graph gives
\begin{equation}\label{eq:global-h-lower}
 H_n=\frac2{N_n}\sum_{i<j}\mathsf h_{ij}
 \succeq
 \frac32\frac2{N_n}\sum_{i<j}P^-_{ij}.
\end{equation}
In particular, $H_n$ is positive semidefinite, and hence $A_n$ is
negative semidefinite.

\begin{proposition}\label{prop:lower-spin-bound}
For every $0\leqslant r\leqslant\lfloor n/2\rfloor$,
\begin{equation}\label{eq:sector-lower-bound}
 H_{n,r}\succeq
 \frac{3r(n-r+1)}{2N_n}I.
\end{equation}
In particular, if $r\geqslant2$, then
\begin{equation}\label{eq:r-ge-two-bound}
 \lambda_{\min}(H_{n,r})\geqslant\frac3n.
\end{equation}
\end{proposition}

\begin{proof}
Since $P^-_{ij}=(I-(ij))/2$, Lemma~\ref{lem:B-scalar} gives
\begin{equation*}\label{eq:Pminus-scalar}
 \left.\frac2{N_n}\sum_{i<j}P^-_{ij}\right|_{\cW_{n,r}}
 =\frac{r(n-r+1)}{N_n}I.
\end{equation*}
Combining this with \eqref{eq:global-h-lower} proves
\eqref{eq:sector-lower-bound}.

For $2\leqslant r\leqslant n/2$,
\begin{equation*}\label{eq:r-product-comparison}
 r(n-r+1)-2(n-1)=(r-2)(n-r-1)\geqslant0.
\end{equation*}
Thus \eqref{eq:sector-lower-bound} is at least
$3(n-1)/N_n=3/n$, which proves \eqref{eq:r-ge-two-bound}.
\end{proof}

\subsection{The highest-spin zero mode}

The highest-spin block has two zero modes.  One is the isolated vector
$\mathbf e_0^{(0)}$.  The other lies in the Jacobi chain
$H_{n,0}^{\circ}$.

\begin{lemma}\label{lem:highest-zero-mode}
The vector
\begin{equation}\label{eq:phi-vector}
 {\phi}=\sum_{q=1}^{n}\phi_q\mathbf e_q^{(0)},
 \qquad
 \phi_q=\sqrt{\frac{3^q\binom nq}{4^n-1}},
\end{equation}
is a unit vector in the kernel of $H_{n,0}^{\circ}$.  This zero
eigenvalue is simple.
\end{lemma}

\begin{proof}
The normalisation follows from
\begin{equation*}\label{eq:phi-normalisation}
 \sum_{q=1}^{n}3^q\binom nq=4^n-1.
\end{equation*}
The ratios of consecutive coordinates are
\begin{equation*}\label{eq:phi-ratios}
 \frac{\phi_{q+1}}{\phi_q}
 =\sqrt{\frac{3(n-q)}{q+1}}
 \quad (1\leqslant q<n),
 \qquad
 \frac{\phi_{q-1}}{\phi_q}
 =\sqrt{\frac{q}{3(n-q+1)}}
 \quad (1<q\leqslant n).
\end{equation*}
Using \eqref{eq:H0-off-diagonal}, the contributions from the adjacent
coordinates to the $q$th row are
\begin{align}
 \frac{\sqrt3\,q\sqrt{(q+1)(n-q)}}{N_n}
 \frac{\phi_{q+1}}{\phi_q}
 &=\frac{3q(n-q)}{N_n}
 &&(1\leqslant q<n),
 \label{eq:phi-upper-contribution}\\
 \frac{\sqrt3\,(q-1)\sqrt{q(n-q+1)}}{N_n}
 \frac{\phi_{q-1}}{\phi_q}
 &=\frac{q(q-1)}{N_n}
 &&(1<q\leqslant n).
 \label{eq:phi-lower-contribution}
\end{align}
For $2\leqslant q\leqslant n-1$, their sum is
$q(3n-2q-1)/N_n$, exactly the diagonal entry in
\eqref{eq:H0-diagonal}.  At $q=1$ and $q=n$, the same identity holds after
omitting the absent neighbouring term, whose corresponding off-diagonal
coefficient is zero.  Hence $H_{n,0}^{\circ}{\phi}=0$.

Every off-diagonal entry of $H_{n,0}^{\circ}$ is non-zero.  A real
symmetric irreducible Jacobi matrix has simple spectrum
\cite[Chapter~1]{Teschl}; therefore its zero eigenvalue is simple.
\end{proof}

The positive spectrum of $H_{n,0}$ is the positive spectrum of
$H_{n,0}^{\circ}$.  Moreover, the local bound
\eqref{eq:sector-lower-bound} shows that
$\lambda_{\min}(H_{n,r})>0$ for every $r\geqslant1$.  Thus
Theorem~\ref{thm:intro-localisation} is precisely the assertion that
\begin{equation*}
 \lambda_{\min}^{+}(H_{n,0}^{\circ})
 <\lambda_{\min}(H_{n,r})
 \qquad (r\geqslant1),
\end{equation*}
and hence that the first positive eigenvalue of $H_n$ occurs in the
highest-spin sector.

\section{The critical Jacobi comparison}\label{sec:critical}

Proposition~\ref{prop:lower-spin-bound} has already separated every
sector with $r\geqslant2$, provided that we can place
$\lambda_{\min}^{+}(H_{n,0}^{\circ})$ below $3/n$.  
The only close competitor is $H_{n,1}$.  For $n\geqslant9$, we shall prove
\begin{equation}\label{eq:critical-goal}
 \lambda_{\min}^{+}(H_{n,0}^{\circ})
 <\frac{3n-5}{N_n}
 <\lambda_{\min}(H_{n,1}).
\end{equation}

The middle quantity in \eqref{eq:critical-goal} is not an arbitrary
interpolation between two numerically observed eigenvalues.  It is adapted
simultaneously to the first two coefficients of the highest-spin chain,
the first pivot of the next-to-highest-spin chain, and the uniform bound
for all remaining sectors.  More precisely,
\begin{equation*}\label{eq:tau-three-identities}
 \bigl(H_{n,0}^{\circ}\bigr)_{2,2}=\frac{2(3n-5)}{N_n},
 \qquad
 \bigl(H_{n,1}\bigr)_{1,1}-\frac{3n-5}{N_n}=\frac{2}{N_n},
 \qquad
 \frac3n-\frac{3n-5}{N_n}=\frac{2}{N_n}.
\end{equation*}
The first identity makes the two-coordinate variational determinant
particularly simple.  After multiplication by $N_n$, the second makes the
initial Sturm pivot equal to $2$, whilst the third preserves a strict
margin below the coarse lower bound for every sector with $r\geqslant2$.
Thus the explicit rational threshold $(3n-5)/N_n$ is compatible with all
three parts of the proof.  Rather than estimating either critical
eigenvalue directly, we shall separate them across this common threshold.

We recall two standard facts.  Let $J$
be an $m\times m$ real symmetric tridiagonal matrix with non-zero
off-diagonal entries, and let
\begin{equation*}\label{eq:sturm-polynomials-general}
 p_0(t)=1,
 \qquad
 p_k(t)=\det\bigl(J^{[k]}-tI_k\bigr).
\end{equation*}
If none of the values $p_k(t)$ vanishes, then the number of eigenvalues of
$J$ below $t$ is the number of sign changes in
$p_0(t),p_1(t),\ldots,p_m(t)$.  Moreover, a real symmetric matrix is
positive definite if and only if all its leading principal minors are
positive.  These are respectively the finite-dimensional Sturm theorem
and Sylvester's criterion; see, for example,
\cite{HornJohnson,Teschl}.

\subsection{A variational upper bound in the highest-spin sector}

\begin{lemma}\label{lem:highest-variational}
For every $n\geqslant6$,
\begin{equation*}\label{eq:gamma0-less-tau}
 \lambda_{\min}^{+}(H_{n,0}^{\circ})<\frac{3n-5}{N_n}.
\end{equation*}
\end{lemma}

\begin{proof}
Let
\begin{equation*}\label{eq:E-subspace}
 E=\Span\{\mathbf e_1^{(0)},\mathbf e_2^{(0)}\}
 \subseteq\Span\{\mathbf e_1^{(0)},\ldots,\mathbf e_n^{(0)}\},
\end{equation*}
and let $\Pi=I-\proj{{\phi}}$
be the orthogonal projection onto ${\phi}^{\perp}$, where ${\phi}$ is the kernel
vector from Lemma~\ref{lem:highest-zero-mode}.  Since ${\phi}\notin E$, the
restriction of $\Pi$ to $E$ is injective.  For $\mathbf v\in E$ and $\mathbf w=\Pi\mathbf v$, we
have
\begin{equation}\label{eq:projected-quadratic}
 \ip{\mathbf w}{H_{n,0}^{\circ}\mathbf w}
 =\ip{\mathbf v}{H_{n,0}^{\circ}\mathbf v},
\end{equation}
because $H_{n,0}^{\circ}{\phi}=0$.

In the ordered basis $(\mathbf e_1^{(0)},\mathbf e_2^{(0)})$, the numerator of the
Rayleigh quotient on $\Pi E$ is represented by
\begin{equation*}\label{eq:M-matrix}
 M=
 \begin{pmatrix}
 \dfrac3n&-\dfrac{\sqrt{6(n-1)}}{N_n}\\[1em]
 -\dfrac{\sqrt{6(n-1)}}{N_n}&
 \dfrac{2(3n-5)}{N_n}
 \end{pmatrix}.
\end{equation*}
The squared norm of $\Pi\mathbf v$ is represented by the Gram matrix
\begin{equation*}\label{eq:G-matrix}
 G=I-\frac1{4^n-1}
 \begin{pmatrix}
 3n&3n\sqrt{\dfrac{3(n-1)}2}\\[0.8em]
 3n\sqrt{\dfrac{3(n-1)}2}&\dfrac92n(n-1)
 \end{pmatrix}.
\end{equation*}
The matrix $G$ is positive definite, again because $\Pi|_E$ is
injective.

We shall show that the smaller generalised eigenvalue of the pair $(M,G)$
is below $(3n-5)/N_n$.  Direct simplification gives
\begin{equation}\label{eq:det-M-tauG}
 \det\!\left(M-\frac{3n-5}{N_n}G\right)
 =\frac{2\bigl(-2\cdot4^n+54n^3-153n^2+105n+2\bigr)}
 {n^2(4^n-1)(n-1)^2}.
\end{equation}
The expression
\begin{equation*}
 2\cdot4^n-54n^3+153n^2-105n-2
\end{equation*}
equals $1404$ at $n=6$.  Moreover,
\begin{align*}
 &\bigl(2\cdot4^{n+1}-54(n+1)^3+153(n+1)^2-105(n+1)-2\bigr)\notag\\
 &\quad-4\bigl(2\cdot4^n-54n^3+153n^2-105n-2\bigr)
 =27n(n-1)(6n-17)>0
\end{align*}
for $n\geqslant3$.  The expression is therefore positive for every
$n\geqslant6$, and the determinant in \eqref{eq:det-M-tauG} is negative.

Since $G\succ0$, the matrix
\begin{equation*}
 G^{-1/2}\left(M-\frac{3n-5}{N_n}G\right)G^{-1/2}
\end{equation*}
is real symmetric and has negative determinant.  It therefore has a
negative eigenvalue.  Thus there is a non-zero coordinate vector
${\xi}\in\mathbb R^2$ for which
\begin{equation*}\label{eq:generalised-Rayleigh-less}
 \frac{{\xi}^{T}M{\xi}}
 {{\xi}^{T}G{\xi}}
 <\frac{3n-5}{N_n}.
\end{equation*}
Let $\mathbf v\in E$ be the vector with coordinate vector
${\xi}$.  The corresponding vector
$\mathbf w=\Pi\mathbf v\in{\phi}^{\perp}$ has the same quotient
by \eqref{eq:projected-quadratic}.  The min--max principle now gives
\begin{equation*}
 \lambda_{\min}^{+}(H_{n,0}^{\circ})<\frac{3n-5}{N_n}. \qedhere
\end{equation*}
\end{proof}

\begin{remark}\label{rem:two-dimensional-certificate}
The proof of Lemma~\ref{lem:highest-variational} is an exact
two-dimensional certificate.  It improves the one-coordinate trial vector
used in the highest-spin estimate of~\cite{KongLiLiu}; the second coordinate
is precisely what detects the separation from the $r=1$ block at the scale
needed here.
\end{remark}

\subsection{A Sturm lower bound in the next-to-highest-spin sector}

We now prove the opposite inequality for $H_{n,1}$.  It is enough to show
that
\begin{equation}\label{eq:shifted-next-spin-positive}
 N_nH_{n,1}-(3n-5)I\succ0.
\end{equation}
By \eqref{eq:H1-diagonal}--\eqref{eq:H1-off-diagonal}, the matrix on the
left is an $(n-1)\times(n-1)$ Jacobi matrix with diagonal entries
\begin{equation*}\label{eq:Kn-diagonal}
 q(3n-2q-1)-(3n-5),
 \qquad 1\leqslant q\leqslant n-1,
\end{equation*}
and upper off-diagonal entries
\begin{equation*}\label{eq:Kn-off-diagonal}
 -\sqrt{3q^3(n-q-1)},
 \qquad 1\leqslant q<n-1.
\end{equation*}
Let
\begin{equation*}\label{eq:Dk-definition}
 D_k=D_k(n)
 =\det\!\left(\bigl(N_nH_{n,1}-(3n-5)I\bigr)^{[k]}\right),
 \qquad D_0=1.
\end{equation*}
The tridiagonal determinant recurrence is
\begin{equation}\label{eq:Dk-recurrence}
 D_1=2,
 \qquad
 D_k=\bigl(k(3n-2k-1)-(3n-5)\bigr)D_{k-1}
 -3(k-1)^3(n-k)D_{k-2}.
\end{equation}

\begin{lemma}\label{lem:next-spin-Sturm}
For every $n\geqslant9$,
\begin{equation*}\label{eq:gamma1-greater-tau}
 \lambda_{\min}(H_{n,1})>\frac{3n-5}{N_n}.
\end{equation*}
\end{lemma}

\begin{proof}
We prove that all leading principal minors of
$N_nH_{n,1}-(3n-5)I$ are positive.
Put
\begin{equation*}\label{eq:s-n-minus-nine}
 s=n-9\geqslant0.
\end{equation*}
Repeated use of \eqref{eq:Dk-recurrence} gives
\begin{align}
 D_1={}&2,\label{eq:D1}\\
 D_2={}&3s+23,\label{eq:D2}\\
 D_3={}&18s^2+204s+586,\label{eq:D3}\\
 D_4={}&162s^3+2493s^2+12396s+19985,\label{eq:D4}\\
 D_5={}&1944s^4+35856s^3+240354s^2+689604s+709082.\label{eq:D5}
\end{align}
Every coefficient displayed here is positive, so $D_1,\ldots,D_5>0$.
The next recurrence step gives the sharper relation
\begin{align}
 D_6-15(n-7)D_5
 ={}&1458s^4+30267s^3+238203s^2+627453s+207499>0.
 \label{eq:D6-base}
\end{align}
In particular, $D_6>0$ and
\begin{equation}\label{eq:D6-ratio}
 \frac{D_6}{D_5}>15(n-7)=3(6-1)(n-7).
\end{equation}

We claim inductively that
\begin{equation}\label{eq:Dk-ratio-claim}
 \frac{D_k}{D_{k-1}}>3(k-1)(n-k-1)
 \qquad (6\leqslant k\leqslant n-1).
\end{equation}
The case $k=6$ is \eqref{eq:D6-ratio}.  Suppose that
$7\leqslant k\leqslant n-1$ and that
\begin{equation}\label{eq:Dkminus1-ratio-hypothesis}
 \frac{D_{k-1}}{D_{k-2}}>3(k-2)(n-k).
\end{equation}
Dividing \eqref{eq:Dk-recurrence} by $D_{k-1}$ and using
\eqref{eq:Dkminus1-ratio-hypothesis}, we obtain
\begin{align}
 \frac{D_k}{D_{k-1}}
 &>k(3n-2k-1)-(3n-5)-\frac{(k-1)^3}{k-2}\notag\\
 &=3(k-1)(n-k-1)+\frac{k-3}{k-2}\notag\\
 &>3(k-1)(n-k-1).
 \label{eq:Dk-ratio-step}
\end{align}
This proves \eqref{eq:Dk-ratio-claim}.

Equations \eqref{eq:D1}--\eqref{eq:D6-base} and
\eqref{eq:Dk-ratio-claim} show that
\begin{equation*}\label{eq:all-D-positive}
 D_k>0
 \qquad (1\leqslant k\leqslant n-1).
\end{equation*}
By Sylvester's criterion, \eqref{eq:shifted-next-spin-positive} holds.
Equivalently,
\begin{equation*}
 \lambda_{\min}(H_{n,1})>\frac{3n-5}{N_n}.\qedhere
\end{equation*}
\end{proof}

\subsection{The remaining finite cases}

The threshold argument now gives the desired comparison for $n\geqslant9$.
Indeed, Lemmas~\ref{lem:highest-variational} and
\ref{lem:next-spin-Sturm} yield
\begin{equation}\label{eq:large-critical-separation}
 \lambda_{\min}^{+}(H_{n,0}^{\circ})
 <\frac{3n-5}{N_n}
 <\lambda_{\min}(H_{n,1}).
\end{equation}
Moreover,
\begin{equation}\label{eq:tau-less-three-over-n}
 \frac{3n-5}{N_n}<\frac3n,
\end{equation}
so Proposition~\ref{prop:lower-spin-bound} separates every $r\geqslant2$.

For $n=5,6,7,8$, we use exact rational thresholds.  Set
\begin{equation}\label{eq:theta-small}
 \theta_5=\frac{19}{2},
 \qquad
 \theta_6=12,
 \qquad
 \theta_7=15,
 \qquad
 \theta_8=18.
\end{equation}
Let
\begin{equation*}\label{eq:L-small}
 L_{n,0}=N_nH_{n,0}^{\circ},
 \qquad
 L_{n,1}=N_nH_{n,1}.
\end{equation*}
The exact leading principal determinants are listed in
Appendix~\ref{app:small}.  Their signs give the following result.

\begin{lemma}\label{lem:small-dimensions}
For $5\leqslant n\leqslant8$,
\begin{equation}\label{eq:small-separation}
 \lambda_{\min}^{+}(H_{n,0}^{\circ})
 <\frac{\theta_n}{N_n}
 <\lambda_{\min}(H_{n,1}).
\end{equation}
Furthermore,
\begin{equation}\label{eq:small-threshold-three-over-n}
 \frac{\theta_n}{N_n}<\frac3n.
\end{equation}
\end{lemma}

\begin{proof}
For $L_{n,0}-\theta_nI$, the Sturm sequences in
Appendix~\ref{app:small} have exactly two sign changes.  Hence
$L_{n,0}$ has exactly two eigenvalues below $\theta_n$.  Its spectrum is
non-negative, and its zero eigenvalue is simple by
Lemma~\ref{lem:highest-zero-mode}.  The next eigenvalue is therefore
$N_n\lambda_{\min}^{+}(H_{n,0}^{\circ})$, which proves the left-hand inequality in
\eqref{eq:small-separation}.

For $L_{n,1}-\theta_nI$, every leading principal minor listed in the
appendix is positive.  Sylvester's criterion gives
$L_{n,1}-\theta_nI\succ0$, proving the right-hand inequality.
Finally, direct inspection of \eqref{eq:theta-small} gives
$\theta_n<3(n-1)$, which is equivalent to
\eqref{eq:small-threshold-three-over-n}.
\end{proof}

We can now complete the representation-theoretic part of the argument.

\begin{theorem}\label{thm:spin-localisation}
Let $n\geqslant5$.  Then
\begin{equation}\label{eq:all-sector-separation}
 \lambda_{\min}^{+}(H_{n,0}^{\circ})<\lambda_{\min}(H_{n,r})
 \qquad
 \left(1\leqslant r\leqslant\left\lfloor\frac n2\right\rfloor\right).
\end{equation}
Consequently, the smallest positive eigenvalue of $H_n=-A_n$ occurs in
the highest-spin sector $\cW_{n,0}$.
\end{theorem}

\begin{proof}
Suppose first that $n\geqslant9$.  The case $r=1$ follows from
\eqref{eq:large-critical-separation}.  If $r\geqslant2$, then
Lemma~\ref{lem:highest-variational} and
\eqref{eq:tau-less-three-over-n}, together with
Proposition~\ref{prop:lower-spin-bound}, give
\begin{equation}\label{eq:large-r-separation}
 \lambda_{\min}^{+}(H_{n,0}^{\circ})
 <\frac{3n-5}{N_n}<\frac3n\leqslant\lambda_{\min}(H_{n,r}).
\end{equation}

Now suppose that $5\leqslant n\leqslant8$.  Lemma~\ref{lem:small-dimensions}
handles $r=1$.  For $r\geqslant2$, equations
\eqref{eq:small-separation}, \eqref{eq:small-threshold-three-over-n} and
\eqref{eq:r-ge-two-bound} give
\begin{equation*}\label{eq:small-r-separation}
 \lambda_{\min}^{+}(H_{n,0}^{\circ})
 <\frac{\theta_n}{N_n}<\frac3n\leqslant\lambda_{\min}(H_{n,r}).
\end{equation*}
This proves \eqref{eq:all-sector-separation} in every case.
\end{proof}

\begin{corollary}\label{cor:A-highest-spin}
For $n\geqslant5$, the largest strictly negative eigenvalue of $A_n$ is
$-\lambda_{\min}^{+}(H_{n,0}^{\circ})$ and occurs in the highest-spin
sector $\cW_{n,0}$.  Every lower-spin sector has largest eigenvalue
strictly smaller than $-\lambda_{\min}^{+}(H_{n,0}^{\circ})$.
\end{corollary}

\begin{proof}
The operator $A_n=-H_n$ is negative semidefinite.  In the highest-spin
sector its largest strictly negative eigenvalue is
$-\lambda_{\min}^{+}(H_{n,0}^{\circ})$, whilst in the $r$th lower-spin
sector its largest eigenvalue is $-\lambda_{\min}(H_{n,r})$.  The
conclusion follows from Theorem~\ref{thm:spin-localisation}.
\end{proof}

Corollary~\ref{cor:A-highest-spin} is precisely the strengthening of the
highest-spin lemma proposed in
\cite[the discussion following Conjecture~2]{KongLiLiu}.

\section{Proof of the complete-graph gap theorem}\label{sec:main-proof}

We now return to the moment operator.  Throughout this section, assume
that $a=5/9$, so that $0\leqslant c\leqslant1/3$ by
Lemma~\ref{lem:c-range}.

\subsection{Independence of the parameter \texorpdfstring{$c$}{c}}

On the sector indexed by $r$, equations
\eqref{eq:restricted-a-five-nine}, \eqref{eq:H-definition} and
\eqref{eq:B-scalar} give
\begin{equation}\label{eq:T-sector}
 \left.\TIG_{2,K_n}\right|_{\cW_{n,r}}
 =I-\frac49H_{n,r}
 -\frac{2cr(n-r+1)}{N_n}I.
\end{equation}
For $r=0$, the last term vanishes.  The two zero modes of $H_{n,0}$ give
the fixed space, and the largest eigenvalue strictly below $1$ in this
sector is
\begin{equation}\label{eq:highest-sector-candidate}
 1-\frac49\lambda_{\min}^{+}(H_{n,0}^{\circ}).
\end{equation}
For $r\geqslant1$, the largest eigenvalue of \eqref{eq:T-sector} is
\begin{equation}\label{eq:lower-sector-candidate}
 1-\frac49\lambda_{\min}(H_{n,r})
 -\frac{2cr(n-r+1)}{N_n}.
\end{equation}
By Theorem~\ref{thm:spin-localisation} and $c\geqslant0$, the value in
\eqref{eq:lower-sector-candidate} is strictly smaller than the value in
\eqref{eq:highest-sector-candidate}.  We have proved the following.

\begin{proposition}\label{prop:restricted-second-eigenvalue}
Let $n\geqslant5$, and suppose that $a=5/9$.  Then
\begin{equation*}\label{eq:restricted-lambda-star}
 \lamstar\!\left(
 \TIG_{2,K_n}\big|_{\cV^{\otimes n}}
 \right)
 =1-\frac49\lambda_{\min}^{+}(H_{n,0}^{\circ}).
\end{equation*}
In particular, this eigenvalue is independent of $c$ and is attained in
the highest-spin sector.
\end{proposition}

\begin{remark}\label{rem:strictness}
The strict comparison in Theorem~\ref{thm:spin-localisation} is stronger
than is needed when $c>0$, because the scalar $B_n$ term then creates an
additional separation.  It is indispensable at $c=0$, which includes the
CNOT gadget and is the least favourable case for excluding lower-spin
competitors.
\end{remark}

\subsection{The complementary subspace}

The restriction to $\cV^{\otimes n}$ does not, by itself, determine the
spectrum of the full moment operator.  We combine the following two
results of Kong, Li, and Liu.

\begin{proposition}[{\cite[Lemma~4.13 and Corollary~4.14]{KongLiLiu}}]\label{prop:KLL-complement}
Let $T^{\mathrm{IG}}_{2,G}$ be the second-moment operator of an ironed
gadget on a graph with $n$ vertices, and let $T^{\mathrm{IG}}_2$ be its
local two-qubit second-moment operator.  If
\begin{equation*}\label{eq:KLL-condition}
 n\geqslant\frac{2}{1+\lambda_{\min}(T^{\mathrm{IG}}_2)},
\end{equation*}
then
\begin{equation*}\label{eq:KLL-lambda-reduction}
 \lamstar(T^{\mathrm{IG}}_{2,G})
 =\lamstar\!\left(
 T^{\mathrm{IG}}_{2,G}\big|_{\cV^{\otimes n}}
 \right).
\end{equation*}
Under the same hypothesis, the spectral gap is determined by this
eigenvalue:
\begin{equation}\label{eq:KLL-gap-reduction}
 \Delta(T^{\mathrm{IG}}_{2,G})
 =1-\lamstar(T^{\mathrm{IG}}_{2,G}).
\end{equation}
\end{proposition}

The local one-qubit Haar projectors annihilate the orthogonal complement
of $\cV$ at either endpoint.  Thus the full local moment operator has the
four eigenvalues of \eqref{eq:local-matrix}, together with zero
eigenvalues.  Lemma~\ref{lem:c-range} therefore gives
\begin{equation*}\label{eq:full-local-min}
 \lambda_{\min}(T^{\mathrm{IG}}_2)\geqslant-\frac13.
\end{equation*}
It follows that
\begin{equation}\label{eq:condition-at-most-three}
 \frac{2}{1+\lambda_{\min}(T^{\mathrm{IG}}_2)}\leqslant3.
\end{equation}
In particular, the hypothesis of Proposition~\ref{prop:KLL-complement}
holds whenever $n\geqslant5$.

\begin{theorem}\label{thm:main}
Let $n\geqslant5$, and let $\TIG_{2,K_n}$ be the second-moment operator of an
ironed gadget with parameter $a=5/9$.  Then
\begin{equation}\label{eq:gap-gamma}
 \Delta(\TIG_{2,K_n})=\frac49\lambda_{\min}^{+}(H_{n,0}^{\circ}),
\end{equation}
where $\lambda_{\min}^{+}(H_{n,0}^{\circ})$ is the first positive eigenvalue of the explicit
Jacobi matrix in \eqref{eq:H0-diagonal}--\eqref{eq:H0-off-diagonal}.
Consequently,
\begin{equation}\label{eq:gap-equality-iSWAP}
 \Delta(\TIG_{2,K_n})
 =\Delta(\TiSWAP_{2,K_n}).
\end{equation}
\end{theorem}

\begin{proof}
By \eqref{eq:condition-at-most-three},
Proposition~\ref{prop:KLL-complement} applies.  Combining it with
Proposition~\ref{prop:restricted-second-eigenvalue} gives
\begin{equation*}\label{eq:full-lambda-star}
 \lamstar(\TIG_{2,K_n})=1-\frac49\lambda_{\min}^{+}(H_{n,0}^{\circ}).
\end{equation*}
The gap reduction \eqref{eq:KLL-gap-reduction} now yields
\eqref{eq:gap-gamma}.

The iSWAP gate has KAK coordinates
$(k_x,k_y,k_z)=(\pi/4,\pi/4,0)$ in the convention
\eqref{eq:KAK}.  Hence
\begin{equation*}\label{eq:iSWAP-xyz}
 (x,y,z)=(-1,-1,1),
\end{equation*}
which gives $a=5/9$ and $c=1/3$.  Formula \eqref{eq:gap-gamma} is
independent of $c$, so it has the same value for every gadget with
$a=5/9$ and, in particular, for iSWAP.  This proves
\eqref{eq:gap-equality-iSWAP}.
\end{proof}

Theorem~\ref{thm:main} is Theorem~\ref{thm:intro-main} and settles
Conjecture~2 of~\cite{KongLiLiu}.

\begin{corollary}\label{cor:standard-gates}
For every $n\geqslant5$, the ironed gadgets associated with iSWAP, the
$B$ gate, and CNOT have equal second-moment spectral gaps on $K_n$.
\end{corollary}

\begin{proof}
Each of these gate families has $a=5/9$; see
\cite[Table~3 and Conjecture~2]{KongLiLiu}.  The conclusion follows from
Theorem~\ref{thm:main}.
\end{proof}


\appendix

\section{Exact Sturm certificates in dimensions five to eight}
\label{app:small}

For completeness, we give the exact data used in
Lemma~\ref{lem:small-dimensions}.  Recall that
\begin{equation*}\label{eq:appendix-L}
 L_{n,0}=N_nH_{n,0}^{\circ},
 \qquad
 L_{n,1}=N_nH_{n,1},
\end{equation*}
and that the thresholds are
\begin{equation*}\label{eq:appendix-theta}
 \theta_5=\frac{19}{2},
 \qquad
 \theta_6=12,
 \qquad
 \theta_7=15,
 \qquad
 \theta_8=18.
\end{equation*}

For $L_{n,0}-\theta_nI$, let
\begin{equation*}\label{eq:Ek-appendix}
 E_k^{(n)}
 =\det\bigl((L_{n,0}-\theta_nI)^{[k]}\bigr),
 \qquad E_0^{(n)}=1.
\end{equation*}
The complete sequences are
\begin{align*}
 n=5:\quad
 &(1,\tfrac52,\tfrac94,-\tfrac{1899}{8},
 -\tfrac{62847}{16},\tfrac{503253}{32}),
 \\
 n=6:\quad
 &(1,3,12,-180,-8208,-102384,1850688),
 \\
 n=7:\quad
 &(1,3,15,-135,-10935,-285525,419175,227174625),
 \\
 n=8:\quad
 &(1,3,18,-54,-11988,-511596,-8372808,
 396756792,24923180304).
\end{align*}
Their sign patterns are respectively
\begin{equation*}\label{eq:E-signs}
 +++--+,
 \qquad
 +++---+,
 \qquad
 +++---++,
 \qquad
 +++----++.
\end{equation*}
Each pattern has exactly two sign changes.  Since no term vanishes, the
finite Sturm theorem shows that $L_{n,0}$ has exactly two eigenvalues below
$\theta_n$.

For $L_{n,1}-\theta_nI$, let
\begin{equation*}\label{eq:Qk-appendix}
 Q_k^{(n)}
 =\det\bigl((L_{n,1}-\theta_nI)^{[k]}\bigr).
\end{equation*}
The sequences of leading principal minors are
\begin{align*}
 n=5:\quad
 &(\tfrac52,\tfrac{69}{4},\tfrac{1041}{8},\tfrac{7833}{16}),
 \\
 n=6:\quad
 &(3,30,414,5076,37260),
 \\
 n=7:\quad
 &(3,36,684,13824,221184,2115072),
 \\
 n=8:\quad
 &(3,42,1026,29484,794772,16036056,206610264).
\end{align*}
Every entry is positive.  Hence $L_{n,1}-\theta_nI$ is positive definite
by Sylvester's criterion.  These rational certificates establish
\begin{equation*}\label{eq:appendix-conclusion}
 N_n\lambda_{\min}^{+}(H_{n,0}^{\circ})<\theta_n<N_n\lambda_{\min}(H_{n,1})
 \qquad (5\leqslant n\leqslant8)
\end{equation*}
without numerical approximation.

\bigskip
\noindent{\bf Acknowledgements.}
Y. L. was supported by the Guangdong Basic and Applied Basic Research Foundation (2026A1515011707).
H.Z. gratefully acknowledges support from the NTU Research Scholarship.


\begin{thebibliography}{99}

\bibitem{BaerHaah}
T. Baer and J. Haah,
\emph{Random unitary circuits with constant spectral gap},
arXiv:2607.20919 [quant-ph], 2026.

\bibitem{BrandaoHarrowHorodecki}
F.G.S.L. Brand\~ao, A.W. Harrow and M. Horodecki,
\emph{Local random quantum circuits are approximate polynomial-designs},
Commun. Math. Phys. \textbf{346} (2016), 397--434.

\bibitem{BrownViola}
W.G. Brown and L. Viola,
\emph{Convergence rates for arbitrary statistical moments of random quantum circuits},
Phys. Rev. Lett. \textbf{104} (2010), 250501.

\bibitem{CaputoLiggettRichthammer}
P. Caputo, T.M. Liggett and T. Richthammer,
\emph{Proof of Aldous' spectral gap conjecture},
J. Amer. Math. Soc. \textbf{23} (2010), 831--851.

\bibitem{DankertCleveEmersonLivine}
C. Dankert, R. Cleve, J. Emerson and E. Livine,
\emph{Exact and approximate unitary \(2\)-designs and their application to fidelity estimation},
Phys. Rev. A \textbf{80} (2009), 012304.

\bibitem{DenerisBermejoBracciaCincioCerezo}
A.E. Deneris, P. Bermejo, P. Braccia, L. Cincio and M. Cerezo,
\emph{Exact spectral gaps of random one-dimensional quantum circuits},
Phys. Rev. A \textbf{112} (2025), 062619.

\bibitem{DinizJonathan}
I.T. Diniz and D. Jonathan,
\emph{Comment on the paper ``Random Quantum Circuits are Approximate
\(2\)-designs''},
Commun. Math. Phys. \textbf{304} (2011), 281--293.

\bibitem{FultonHarris}
W. Fulton and J. Harris,
\emph{Representation Theory: A First Course},
Graduate Texts in Mathematics, vol. 129, Springer-Verlag, New York, 1991.

\bibitem{GoodmanWallach}
R. Goodman and N.R. Wallach,
\emph{Symmetry, Representations, and Invariants},
Graduate Texts in Mathematics, vol. 255, Springer, Dordrecht, 2009.

\bibitem{GrossAudenaertEisert}
D. Gross, K. Audenaert and J. Eisert,
\emph{Evenly distributed unitaries: on the structure of unitary designs},
J. Math. Phys. \textbf{48} (2007), 052104.

\bibitem{Haferkamp}
J. Haferkamp,
\emph{Random quantum circuits are approximate unitary \(t\)-designs in depth \(O(nt^{5+o(1)})\)},
Quantum \textbf{6} (2022), 795.

\bibitem{HaferkampHunterJones}
J. Haferkamp and N. Hunter-Jones,
\emph{Improved spectral gaps for random quantum circuits: large local dimensions and all-to-all interactions},
Phys. Rev. A \textbf{104} (2021), 022417.

\bibitem{HarrowLow}
A.W. Harrow and R.A. Low,
\emph{Random quantum circuits are approximate \(2\)-designs},
Commun. Math. Phys. \textbf{291} (2009), 257--302.

\bibitem{HarrowMehraban}
A.W. Harrow and S. Mehraban,
\emph{Approximate unitary \(t\)-designs by short random quantum circuits using nearest-neighbor and long-range gates},
Commun. Math. Phys. \textbf{401} (2023), 1531--1626.

\bibitem{HornJohnson}
R.A. Horn and C.R. Johnson,
\emph{Matrix Analysis}, 2nd ed.,
Cambridge University Press, Cambridge, 2013.

\bibitem{HunterJones}
N. Hunter-Jones,
\emph{Unitary designs from statistical mechanics in random quantum circuits},
arXiv:1905.12053 [quant-ph], 2019.

\bibitem{KhanejaBrockettGlaser}
N. Khaneja, R. Brockett and S.J. Glaser,
\emph{Time optimal control in spin systems},
Phys. Rev. A \textbf{63} (2001), 032308.

\bibitem{KongLiLiu}
L. Kong, Z. Li and Z.-W. Liu,
\emph{Convergence efficiency of quantum gates and circuits},
arXiv:2411.04898 [quant-ph], 2024.

\bibitem{LiangZhu2026}
Y. Liang and H. Zhu,
\emph{iSWAP maximises the second-moment spectral gap in random quantum circuits},
in preparation, 2026.

\bibitem{Sagan}
B.E. Sagan,
\emph{The Symmetric Group: Representations, Combinatorial Algorithms, and Symmetric Functions},
2nd ed., Graduate Texts in Mathematics, vol. 203, Springer-Verlag, New York, 2001.

\bibitem{Teschl}
G. Teschl,
\emph{Jacobi Operators and Completely Integrable Nonlinear Lattices},
Mathematical Surveys and Monographs, vol. 72, American Mathematical Society, Providence, RI, 2000.

\bibitem{YadaSuzukiMitsuhashiYoshioka}
T. Yada, R. Suzuki, Y. Mitsuhashi and N. Yoshioka,
\emph{Non-Haar random circuits form unitary designs as fast as Haar random circuits},
Phys. Rev. Lett. \textbf{136} (2026), 030401.

\bibitem{ZhangValaSastryWhaley}
J. Zhang, J. Vala, S. Sastry and K.B. Whaley,
\emph{A geometric theory of non-local two-qubit operations},
Phys. Rev. A \textbf{67} (2003), 042313.

\end{thebibliography}
\end{document}